\documentclass[a4paper,UKenglish,cleveref, autoref, thm-restate, numberwithinsect]{lipics-v2021}

\pdfoutput=1 %
\hideLIPIcs  %

\title{Context-Free Languages of String Diagrams} %

\author{Matt Earnshaw}{Department of Software Science, Tallinn University of Technology, Estonia \and \url{https://ioc.ee/~matt}}{matt@earnshaw.org.uk}{https://orcid.org/0000-0001-8236-2811}{}%

\author{Mario Román}{Department of Computer Science, University of Oxford, United Kingdom \and \url{https://mroman42.github.io}}{mroman42@gmail.com}{https://orcid.org/0000-0003-3158-1226}{}

\authorrunning{M. Earnshaw and M. Román} %

\Copyright{Matt Earnshaw and Mario Román} %



\category{} %

\relatedversion{} %

\funding{This research was supported by the ESF funded Estonian IT Academy research measure (project 2014-2020.4.05.19-0001) and the Estonian Research Council grant PRG1210.}%

\acknowledgements{We would like to thank Noam Zeilberger, Paul-André Melliès, Pawe\l{} Soboci\'{n}ski, and Chad Nester for helpful discussions, and Tobias Heindel for allowing us to believe it was possible.}

\nolinenumbers %

\EventEditors{John Q. Open and Joan R. Access}
\EventNoEds{2}
\EventLongTitle{42nd Conference on Very Important Topics (CVIT 2016)}
\EventShortTitle{CVIT 2016}
\EventAcronym{CVIT}
\EventYear{2016}
\EventDate{December 24--27, 2016}
\EventLocation{Little Whinging, United Kingdom}
\EventLogo{}
\SeriesVolume{42}
\ArticleNo{23}

\usepackage[leftmargin=1em,rightmargin=1ex,vskip=1ex]{quoting}
\usepackage{amsfonts}
\usepackage{amsthm}
\usepackage{amsmath}
\usepackage{cleveref} %
\usepackage{caption}
\usepackage{mathtools} %
\usepackage{tikz}
\usepackage{tikz-cd}
\usepackage{svg}
\usepackage{float}
\usepackage{subcaption}
\usepackage{scalerel}
\usepackage{mathrsfs}
\usepackage{bussproofs}
\usepackage{hyphenat}
\usepackage{adjustbox}
\usepackage{stackengine}
\usepackage[utf8]{inputenc}
\usepackage[T1]{fontenc}
\usepackage{stmaryrd}
\usepackage[notion,xcolor,hyperref,quotation]{knowledge}
\knowledgeconfigure{protect quotation={tikzcd}}

\usepackage{unicode-mario}
\usepackage{macros}
\usepackage{mathpartir}

\knowledge{notion}
 | regular monoidal grammar
 | Regular monoidal grammars
 | regular monoidal grammars

\knowledge{notion} 
| symmetric multigraph
| symmetric multigraphs
| Symmetric multigraph
| Symmetric multigraphs

\knowledge{notion}
| morphism of symmetric multigraphs

\knowledge{notion}
 | context-free monoidal grammar
 | context-free monoidal grammars
 | Context-free monoidal grammars
 | Context-free monoidal grammar

\knowledge{notion}
 | finitary

\knowledge{notion}
| polygraph
| polygraphs

\knowledge{notion}
| regular monoidal
| regular monoidal language
| regular monoidal languages
| Regular monoidal languages

\knowledge{notion}
 | monoidal language
 | monoidal languages
 | language of string diagrams
 | languages of string diagrams

\knowledge{notion}
 | regular grammar
 | regular grammars
 | Regular grammars

\knowledge{notion}
 | multigraph
 | multigraphs
 | Multigraph
 | Multigraphs

\knowledge{notion}
 | multicategory of spliced arrows
 | spliced arrows
 | spliced arrow

\knowledge{notion}
 | free strict monoidal category
 | free monoidal category

\knowledge{notion}
 | contour

\knowledge{notion}
 | string diagrams
 | string diagram
 | String diagram
 | String diagrams

\knowledge{notion}
 | context-free languages of string diagrams
 | context-free monoidal language
 | context-free monoidal languages
 | Context-free monoidal languages
 
\knowledge{notion}
 | monoidal categories
 | monoidal category
 | \MonCat
 | strict monoidal category
 | Strict monoidal category
 | strict monoidal categories

\knowledge{notion}
 | multicategory
 | multicategories
 | Multicategory
 | Multicategories

\knowledge{notion}
 | multicategory of raw optics
 | raw optics
 | raw optic

\knowledge{notion}
 | optical contour
 | Optical contour
 | optical contours

\knowledge{notion}
 | monoidal functor
 | strict monoidal functors
 | strict monoidal functor
 | Strict monoidal functor

\knowledge{notion}
 | morphism of multigraphs

\knowledge{notion}
 | multifunctor
 | multifunctors
 
\knowledge{notion}
 | raw representative
 
\knowledge{notion}
 | regular representation
 | regular representative
 
\knowledge{notion}
 | morphism of polygraphs

\knowledge{notion}
 | sector
 | sectors

\knowledge{notion}
 | derivation
 | derivations

\knowledge{notion}
 | operation
 | operations
 | Operations

\knowledge{notion}
 | monoidal automaton
 | monoidal automata
 | non-deterministic monoidal automata
 | Non-deterministic monoidal automata
 | non-deterministic monoidal automaton

\knowledge{notion}
 | generator
 | generators
 | Generator
 | Generators

 \knowledge{notion}
 | diagram context
 | diagram contexts
 | Diagram contexts
 | holed diagram
 | holed diagrams
 | Holed diagrams
 | Holed diagram

\knowledge{notion}
 | cartesian monoidal category
 | cartesian monoidal categories
 | Cartesian category
 | cartesian category

\knowledge{notion}
 | hypergraph categories
 | Hypergraph category
 | free hypergraph category
 | hypergraph category

\knowledge{notion}
 | symmetric multifunctor

\knowledge{notion}
 | symmetric multicategory
 
\knowledge{notion}
 | free multicategory

\knowledge{notion}
 | symmetric monoidal category
 | symmetric monoidal categories
 | symmetric strict monoidal category
 
\usepackage[disable,color=yellow]{todonotes}

\allowdisplaybreaks

\makeatletter
\renewcommand{\smalloplus}{\mathbin{\mathpalette\make@small\oplus}}
\renewcommand{\smallotimes}{\mathbin{\mathpalette\make@small\otimes}}

\newcommand{\make@small}[2]{%
  \vcenter{\hbox{%
    $\m@th\ifx#1\displaystyle\scriptstyle\else\ifx#1\textstyle\scriptstyle
     \else\scriptscriptstyle\fi\fi#2$%
  }}%
}
\makeatother

\makeatletter
\def\@copyrightspace{\relax}
\makeatother

\begin{document}
\maketitle

\begin{abstract}
    We introduce context-free languages of morphisms in monoidal categories, extending recent work on the categorification of context-free languages, and regular languages of string diagrams. Context-free languages of string diagrams include classical context-free languages of words, trees, and hypergraphs, when instantiated over appropriate monoidal categories. We prove a representation theorem for context-free languages of string diagrams: every such language arises as the image under a monoidal functor of a regular language of string diagrams.
\end{abstract}

\section{Introduction} \label{sec:intro}%
\emph{Monoids} are the classical algebraic home of formal languages, but a long line of research beginning in the 60s has sought to extend the tools and concepts of language theory to other algebraic structures, such as trees \cite{10.5555/267871.267872,10.1145/800169.805428}, traces \cite{booktraces}, hypergraphs \cite{courcelle_engelfriet_2012,habel,rozenberg1997}, models of algebraic theories \cite{EILENBERG1967452,Thatcher1968}, algebras for monads \cite{blumensath2021algebraic,mmso}, and categories \cite{7174900,TILSON198783}.

Categories are ``monoids with many objects'', and passing from the theory of context-free languages in monoids to the theory of context-free languages in categories has been the subject of recent work by \MZ{} \cite{mellies2022parsing,mellies2}. This novel structural point of view suggests a natural generalization to categories with additional structure. Here, we pursue this idea for \emph{monoidal categories}.
On the one hand, strict monoidal categories are \emph{two-dimensional monoids}, and so a natural step from their one-dimensional counterpart. On the other hand, they have a natural graphical syntax, \emph{string diagrams}, providing a fresh approach to languages of graphs.

A vast literature has explored language theory in various algebras of graphs, culminating in the celebrated results of Courcelle \cite{courcelle_engelfriet_2012}. Our point of departure is the claim that many graphical notions can be naturally viewed as \emph{morphisms in monoidal categories}; that is, monoidal categories provide a suitable algebraic framework for graphical formal languages. This manuscript pursues this idea in the context of recent work in the foundations of language theory which takes a structural approach to \emph{context-freeness}. Ultimately, this line of work seeks to unify the various generalizations of context-free languages, and identify reusable tools for reasoning about them.

\subsection{Languages of string diagrams}
Monoidal categories have an intuitive, sound and complete graphical syntax: \emph{string diagrams}. %
String diagrams resemble graphical languages commonly found in engineering and science, and indeed, they allow us to reason about Markov kernels \cite{fritz:synthetic}, linear algebra \cite{bonchi:graphicalaffinealgebra2019}, or quantum processes \cite{abramsky2009categorical}. In computer science, they provide foundations for visual programming \cite{jeffrey1997premonoidal,kuhail2021characterizing}.

The use of string diagrams as a syntax in these various domains suggests the need for a corresponding theory of string diagrams as a formal language. This is one aim of recent work on languages of string diagrams or \emph{monoidal languages}, such as that elaborated by \Sobocinski{} and the first author \cite{earnshaw22,earnshaw_et_al:LIPIcs.MFCS.2023.43}, who introduced the class of \emph{regular monoidal languages}. A monoidal language in this sense is simply a subset of morphisms in a strict monoidal category, just as a classical formal language is a subset of a monoid. In this work, we introduce a natural class of \emph{context-free} monoidal languages, which capture various extended notions of context-free language found in the computer science literature.

\subsection{Context-free languages over categories}
Our main point of reference in this paper is the recent work of \MZ{} \cite{mellies2022parsing,mellies2}. This work is a thoroughgoing refashioning of the theory of context-free languages from a ``fibrational'' point of view. \MZ{} demonstrate that it is natural and fruitful to consider context-free languages over arbitrary \emph{categories}. They introduce an adjunction between \emph{splicing} (introducing gaps or contexts in terms) and \emph{contouring} (linearizing derivation trees), and use it to give a novel conceptual proof of the \CS{} representation theorem.

\MZ{} provide an ample supply of examples of context-free languages in categories, such as context-free languages of runs over an automaton, languages with an explicit end-of-input marker, multiple context-free grammars \cite{SEKI1991191} and a grammar of series-parallel graphs. However, it is less clear how notions such as context-free grammars of trees and hypergraphs fit into this framework. In this paper, we show how this can be accomplished by adapting the machinery of \MZ{} to the wider setting of monoidal categories and their string diagrams. This generalization is non-trivial, and sheds light on the intriguing differences between languages of string diagrams and classical languages. In particular, our two-dimensional version of the \CS{} representation theorem says that every context-free language of string diagrams is the image under a monoidal functor of a regular language of string diagrams: no intersection of context-free and regular languages is necessary.

\paragraph*{Related work} The representation of context-free grammars as certain morphisms of multigraphs was introduced by Walters in a short paper \cite{walters1989}. A similar type-theoretical version of this idea was also introduced by De Groote \cite{de-groote-2001-towards}. As discussed more extensively above, this idea was taken up and substantially refined by \MZ{}, first in a conference paper \cite{mellies2022parsing} and later in an extended version \cite{mellies2}.

A different notion of context-free families of string diagrams has been introduced by Zamdzhiev \cite{zamdzhiev2016}. There, string diagrams are defined combinatorially as \emph{string graphs}, and context-free families are then generated by B-edNCE graph grammars \cite{rozenberg1997}. Though similar, the resulting notion is not directly comparable to ours.
Here, we use the native algebra of monoidal categories and their multicategories of contexts to define and investigate languages. %

Finally, Heindel's abstract \cite{Heindel2017TheCT} claims a proof of a \CS{} theorem for morphisms in symmetric monoidal categories, but the work described in this abstract was never published. Our development is quite different from that outlined in Heindel's abstract. We prove a stronger representation theorem that does not require an intersection of languages; we work without the assumption of symmetry; and we generalize the categorical machinery of \MZ{}.

\paragraph*{Contributions}
We introduce "context-free languages of string diagrams" (\Cref{def:cfml}) and show that they include a wide variety of examples in the computer science literature including context-free languages of trees and hypergraphs.
We introduce the category of "raw optics" (\Cref{def:rawoptics}) over a "monoidal category", and its left adjoint, the "optical contour" (\Cref{def:opticalContour,thm:opticadjoint}).
We use this machinery to prove a representation theorem for "context-free monoidal languages" (\Cref{thm:main}), also relating them to previous work on "regular monoidal languages".

\section{Preliminaries} \label{sec:prelims} %
In this paper, we define context-free grammars as particular \emph{morphisms}. This point of view, while perhaps unfamiliar, is simple and powerful. It suggests natural generalizations of context-free grammars, such as we will pursue in the main body of the paper, and new conceptual tools for reasoning about them. This idea is not original to us; its roots go back to Walters \cite{walters1989}, with recent refinement and extension by \MZ{} \cite{mellies2022parsing,mellies2}. In this section, we introduce these ideas. We also introduce some background on monoidal categories and their string diagrams, which will be needed for our extension of context-free grammars to this realm.

\subsection{Context-free languages in free monoids and other categories}
We introduce the definition of context-free grammar as a morphism of certain \emph{multigraphs}. "Multigraphs" (or \emph{species} in the work of Melliès and Zeilberger \cite{mellies2}) are a kind of graph in which edges have a \emph{list} of sources and single target. It is often helpful to think of a "multigraph" as a \emph{signature}, specifying a set of typed operations. Note that this is a different use of the term \emph{multigraph} from that specifying graphs allowing multiple parallel edges.

\begin{definition}
  A ""multigraph"" $M$ is a set $S$ of \emph{sorts}, and sets $M(X_1,...,X_n;Y)$ of generating ""operations"" (or \emph{multimorphisms}), for each pair of a list of sorts $X_1,...,X_n$ and a sort $Y$. A multigraph is \emph{finite} if sorts and operations are finite sets. A ""morphism of multigraphs"" is given by a function $f$ on sorts and functions $M(X_1,...,X_n;Y) \to M(fX_1, ..., fX_n; fY)$ between sets of operations.
\end{definition}

"Multigraphs" freely generate "multicategories", also known as \emph{operads} (though this term sometimes refers only to the single-sorted, symmetric case). See Leinster \cite{leinster04} for a comprehensive reference on "multicategories". The free multicategory $\fmult{M}$ over a "multigraph" $M$ has as multimorphisms $\fmult{M}(X_1, ..., X_n; Y)$ the ``trees'' rooted at $Y$, with open leaves $X_1, ..., X_n$, that one can build by ``plugging together'' operations in $M$. We call closed trees, i.e. nullary multimorphisms $d \in \fmult{M}(; Y)$, ""derivations"". Every multicategory $𝕄$ has an underlying "multigraph", denoted $\U{𝕄}$, given by forgetting identities and composition.

Every rule in a context-free grammar is of the form $R \to w_1R_1...R_nw_n$, where $R,R_i$ are non-terminals, and $w_i$ are (possibly empty) words over an alphabet $Σ$. The insight of \MZ{} \cite{mellies2} is that this data may be arranged as an operation $R_1,...,R_n \to R$ in a multigraph over an n-ary operation $w_1-...-w_n$ called a \emph{spliced word}: a word with $n$ gaps, as in \Cref{fig:cfgmz}. We introduce the multicategory of \emph{spliced arrows in a category}.

\begin{definition}[\MZ{} \cite{mellies2}] \label{defn:mzsplice}
  The ""multicategory of spliced arrows"", $\Splice{\mathbb{C}}$, over a category $ℂ$, contains, as objects, pairs of objects of $ℂ$, denoted as $\biobj{A}{B}$. Its multimorphisms are morphisms of the original category, but with $n$ ``gaps'' or ``holes'', into which other morphisms (with holes) may be \emph{spliced}. More precisely, the multimorphisms of $\Splice{\mathbb{C}}$ are given by:
  \vspace*{-2mm} $$\Splice{\mathbb{C}}(\biobj{A_1}{B_1}, \dots, \biobj{A_n}{B_n}; \biobj{X}{Y}) := \mathbb{C}(X;A_1) \times \prod_{i=1}^{n-1} \mathbb{C}(B_i; A_{i+1}) \times \mathbb{C}(B_n;Y).$$
    By convention, nullary multimorphisms are morphisms of $ℂ$, that is $\Splice{\mathbb{C}}(;\biobj{X}{Y}) := \mathbb{C}(X;Y)$. The identity is given by a pair of identities of the original category, multicategorical composition is derived from the composition of the original category.
\end{definition}

\begin{figure}
  \centering
  \includegraphics[scale=0.5]{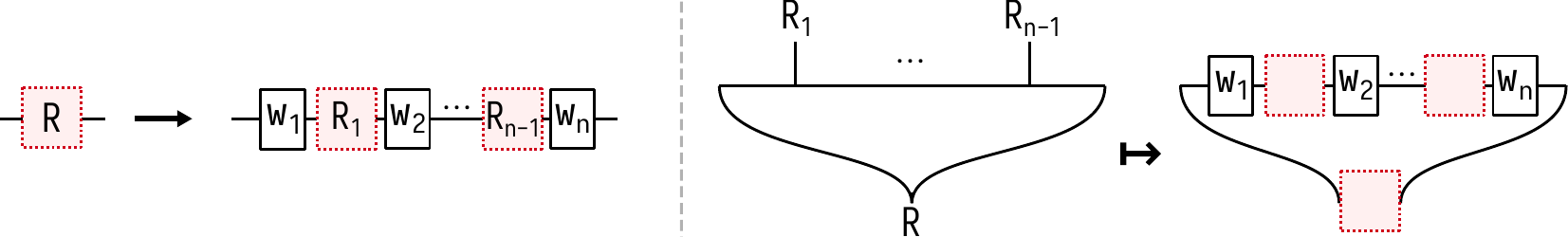}
  \caption{(Left) Generic form of a context-free rule. (Right) Context-free rules as a morphism of multigraphs into spliced arrows; here, spliced arrows in a monoid.}
  \label{fig:cfgmz}
\end{figure}

We can now present a context-free grammar in terms of a morphism of "multigraphs" from a multigraph of non-terminals to the underlying "multigraph" of "spliced arrows", as in \Cref{fig:cfgmz}.

\begin{definition}[\MZ{} \cite{mellies2}] \label{defn:mzcfg}
  A context-free grammar of morphisms in a category $ℂ$ is a morphism of "multigraphs" $G → |\Splice{ℂ}|$ and a sort $S$ in $G$ (the start symbol).
\end{definition}

By the free-forgetful adjunction between "multicategories" and "multigraphs", morphisms $\phi : G \to |\Splice{ℂ}|$ and morphisms of multicategories (or multifunctors) $\widehat{\phi} : \fmult{G} \to \Splice{ℂ}$ are in bijection. This allows for a slick definition of the language of a grammar.

\begin{definition}[\MZ{} \cite{mellies2}] \label{defn:mzcfglang}
  Let $\mathcal{G} = (\phi : G → |\Splice{ℂ}|, S)$ be a context-free grammar of morphisms in $ℂ$. The language of $\mathcal{G}$ is given by the image of the set of "derivations" $\fmult{G}(;S)$ under the multifunctor $\widehat{\phi}$.
\end{definition}

When $ℂ$ is a finitely generated free monoid considered as a one-object category, then context-free grammars over $ℂ$ correspond precisely to the classical context-free grammars. %

An important realization of \MZ{} is that the operation of forming the "multicategory of spliced arrows" in $ℂ$ has a left adjoint. That is, every multicategory gives rise to a category called the ""contour"" of $𝕄$, and this contouring operation is left adjoint to splicing. We refer to their paper for more details \cite[Section 3.2]{mellies2}. Contours give a conceptual replacement for Dyck languages in the classical theory of context-free languages: they linearize the shape of derivation trees.

In \Cref{sec:contour}, we define a new contour of "multicategories" which we call the \emph{optical contour}; we shall use it to prove a representation theorem for languages of string diagrams (\Cref{thm:main}), inspired by generalized \CS{} representation theorem proved by \MZ{}.

\subsection{Monoidal categories, their string diagrams and languages}
In this paper, we will mostly be concerned with monoidal categories presented by generators and equations between the string diagrams built from these generators. Generators are given by \emph{polygraphs}.

\begin{definition} \label{defn:polygraph}
  A ""polygraph"" $Γ$ is a set $S_Γ$ of \emph{sorts}, and sets $Γ(S_1⊗...⊗S_n; T_1⊗...⊗T_m)$ of ""generators"" for every pair of lists $S_i, T_j$ of sorts.
  A polygraph is \emph{finite} if sorts and generators are finite sets.
  A ""morphism of polygraphs"" is a function $f$ on sorts and functions $Γ(S_1⊗...⊗S_n; T_1⊗...⊗T_m) \to Γ(fS_1⊗...⊗fS_n; fT_1⊗...⊗fT_m)$ between generators.
\end{definition}

\begin{figure}[h]
  \centering
  \includegraphics[width=0.9\textwidth]{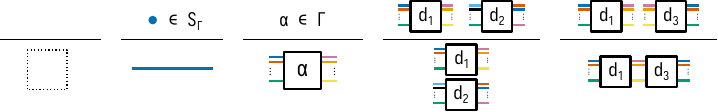}
  \caption{The free strict monoidal category over a polygraph $Γ$ has set of objects $S_Γ^{*}$ and morphisms string diagrams given inductively over the generators of $Γ$ as above, quotiented by isotopy. The leftmost rule denotes the empty diagram. We use colours here to indicate sorts. In string diagrammatic syntax, the usual equations required for term syntax, such as associativity of the tensor product, hold automatically.}
  \label{fig:string-diagrams}
\end{figure}

For a "generator" $γ$ of arity $S_1⊗...⊗S_n$ and coarity $T_1⊗...⊗T_m$ we write $γ : S_1 ⊗ ... ⊗ S_n \to T_1 ⊗ ... ⊗ T_m$. When $S$ is single-sorted, we use natural numbers for the arities and coarities; this case will cover most of the examples in the following. As a string diagrams, we depict generators as boxes with strings on the left and right for their arities and coarities (\Cref{fig:string-diagrams}).

\begin{proposition}
  String diagrams with "generators" in a "polygraph" construct a "monoidal category" (\Cref{fig:string-diagrams}). The monoidal category of string diagrams over a polygraph is the ""free strict monoidal category"" over the polygraph \cite{joyal91,selinger2011}. Every monoidal category is equivalent to a strict one. In particular, string diagrams are sound and complete for monoidal categories.
\end{proposition}

We shall need to impose equations between string diagrams, such as in defining "symmetric monoidal categories", "cartesian monoidal categories" and "hypergraph categories". To this end, we introduce the following notion of presentation.

\begin{definition}\label{defn:finpres}
  A finite presentation of a "strict monoidal category" consists of a finite polygraph of generators, $𝓟$, and a finite polygraph of equations, $𝓔$, with projections for the two sides of each equation, $l,r ፡ 𝓔 → \U{\fpro{𝓟}}$. The "strict monoidal category" presented by $(𝓟,𝓔,l,r)$, is defined as the free "strict monoidal category" generated by $𝓟$ and quotiented by the equations in $𝓔$; in other words, the equalizer of the two projections $l^{\ast}, r^{\ast} ፡ \fpro{𝓔} → \fpro{𝓟}$.
\end{definition}

For the soundness and completeness of string diagrams, see Joyal and Street \cite{joyal91}. For a survey of string diagrams for monoidal categories, see Selinger \cite{selinger2011}. For an introduction to the practical use of string diagrams, see Vicary and Heunen \cite{heunen2019categories} or Hinze and Marsden \cite{hinze}.

\begin{definition}
A ""monoidal language"" or \emph{language of string diagrams} is a subset of morphisms in a "strict monoidal category". 
\end{definition}

\section{Regular Monoidal Languages} \label{sec:regular}
Before introducing "context-free monoidal languages", we introduce the \emph{regular} case, which shall play an important role in \Cref{sec:monoidal-cs}. "Regular monoidal languages" were introduced by \Sobocinski{} and the first author \cite{earnshaw22,earnshaw_et_al:LIPIcs.MFCS.2023.43}, following earlier work of Bossut and Heindel \cite{bossut,Heindel2017AMT}. They are defined by a simple automaton model, reminiscent of tree automata. In a regular monoidal language, the \emph{alphabet} is given by a finite "polygraph". %

\begin{definition} \label{defn:ndmonaut}
  A ""non-deterministic monoidal automaton"" comprises:
  a finite "polygraph" $\Gamma$ (the alphabet);
  a finite set of states $Q$;
  for each generator $\gamma : n → m$ in $\Gamma$, a transition function $\Delta_\gamma : Q^n \to \mathscr{P}(Q^m)$;
  and initial and final state vectors $i, f \in Q^{*}$.
  \end{definition}

\begin{example}
  Classical non-deterministic finite state automata arise as "monoidal automata" over single-sorted "polygraphs" in which every "generator" has arity and coarity 1.
  Bottom-up regular tree automata \cite{10.5555/267871.267872} arise precisely from "monoidal automata" over single-sorted "polygraphs" in which every generator has coarity $1$ and arbitrary arity, with initial state the empty word and final state a singleton. %
\end{example}

A finite state automaton over an alphabet $\Sigma$ accepts elements of the free monoid $\Sigma^{*}$. A "monoidal automaton" over a "polygraph" $Γ$ accepts morphisms in the "free monoidal category" $\fpro{Γ}$ over $\Gamma$. Let us see some examples before giving the formal definition of the accepted language. We depict the transitions of a "monoidal automaton" as elements of a "polygraph" with strings labelled by states, and generators labelled by the corresponding element of $Γ$.

\begin{example}
  Consider the following "polygraph" containing "generators" (left, below) for an opening and closing parenthesis, and the "monoidal automaton" over this polygraph with $Q = \{S,M\}$, $i = f = S$, and transitions shown below, centre. An accepting run over this automaton is shown below, right. The string diagram accepted by this run is what we obtain by erasing the states from this picture.
  \vspace*{-3mm}
  \begin{figure}[H]
      \centering
      \includegraphics[scale=0.5]{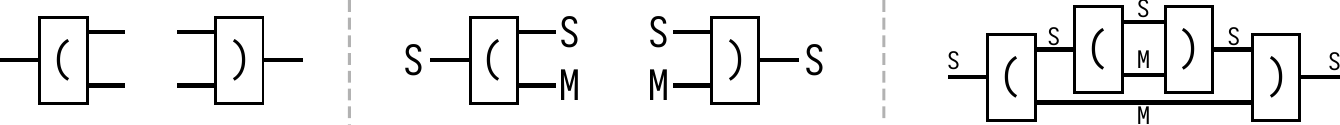}
    \end{figure}
    \vspace*{-4mm}
    It is clear that the language accepted by this automaton is exactly the ``balanced parentheses'', but note that this is not a language of \emph{words}, since we use an extra string to keep track of opening and closing parentheses. This principle will play an important role in our representation theorem in \Cref{sec:monoidal-cs}. Roughly speaking, this extra wire arises from the "optical contour" of a string language of balanced parentheses.
\end{example}

  \begin{example}
  In the field of DNA computing, Rothemund, Papadakis and Winfree demonstrated self-assembly of Sierpiński triangles from \emph{DNA tiles} \cite{10.1371/journal.pbio.0020424}. \Sobocinski{} and the first author \cite{earnshaw22} showed how to recast the tile model as a "regular monoidal language" over a "polygraph" containing two tile generators (white and grey), along with start and end generators, as in \Cref{fig:sierpinski}. Note that the start (end) generators have arity (coarity) 0, and hence effect a transition from (to) the empty word of states.
\end{example}

\begin{figure}[h]
  \centering\includegraphics[width=0.65\columnwidth]{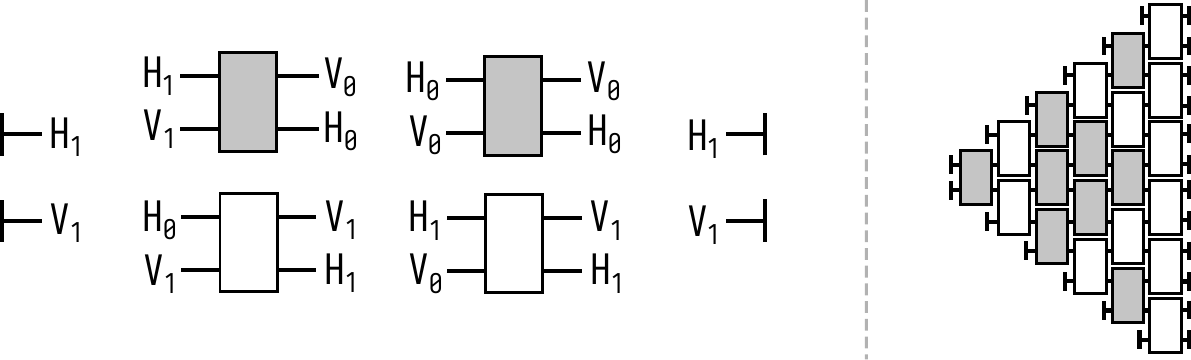}
  \caption{Transitions for the Sierpiński "monoidal automaton" (left) and an element of the language (right). The initial and final states are the empty word.}
  \label{fig:sierpinski}
\end{figure}

  Transitions of a "non-deterministic monoidal automaton" over $Γ$ extend inductively to string diagrams in $\fpro{Γ}$, giving functions $\hat{\delta}_{n,m} : Q^n × \fpro{Γ}(n,m) \to \mathscr{P}(Q^m)$ (\Cref{ax:transition}).

\begin{definition}
 A string diagram $s : n \to m$ in the "free monoidal category" $\fpro{Γ}$ over $Γ$ is in the language of a "non-deterministic monoidal automaton" $(\{\Delta_\gamma\}_{\gamma \in \Gamma}, i, f)$ if and only if $f \in \hat{\delta}_{n,m}(i, s)$. %
\end{definition}

\begin{definition} \label{defn:regmonlang}
  A "monoidal language" $L$ is a ""regular monoidal language"" if and only if there exists a "non-deterministic monoidal automaton" accepting $L$.
\end{definition}

The data of a "monoidal automaton" is a equivalent to a morphism of finite "polygraphs", which we call a \emph{"regular monoidal grammar"}, following Walters' \cite{walters1989} use of the term \emph{grammar} when data is presented as \emph{fibered} over an alphabet, and automata when the alphabet \emph{indexes} transitions as in \Cref{defn:ndmonaut}. We shall use this convenient presentation in the following.

\begin{definition} \label{defn:reg-mon-gram}
A ""regular monoidal grammar"" is a morphism of finite "polygraphs" $\psi : ℚ \to \Gamma$, equipped with finite initial and final sorts $i,f \in S_ℚ^{*}$. The morphisms in $\fpro{ℚ}(i,f)$ are derivations in the grammar, and their image under the free "monoidal functor" $\fpro{\psi}$ gives the language of the grammar; a subset of morphisms in $\fpro{Γ}$.
\end{definition}

\begin{restatable}[]{proposition}{autgram} \label{prop:reg-mon-as-morphism}
For every "non-deterministic monoidal automaton" there is a "regular monoidal grammar" with the same language, and vice-versa.
\end{restatable}

Not every "monoidal language" is a "regular monoidal language". The following is an example.

\begin{proposition} \label{prop:unbraid-reg}
  Let $Γ$ be the "polygraph" containing two generators: one for ``over-braiding'' \includegraphics[scale=0.27]{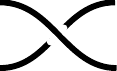} and one for ``under-braiding'' \includegraphics[scale=0.27]{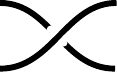}. The language of \emph{unbraids} on two strings over $Γ$, i.e. diagrams equivalent under planar isotopy to untangled strings, is not a "regular monoidal language".
\end{proposition}
\begin{proofsketch}
  We can use the pumping lemma for "regular monoidal languages" (\Cref{lem:pump}), with $k=2$. The argument is analogous to that for classical languages of balanced parentheses: every over-braiding or under-braiding must be eventually balanced with its opposite.
\end{proofsketch}

In the next section, we introduce \emph{context-free monoidal languages} and we shall see that unbraids fall in this class. In \Cref{sec:monoidal-cs} we prove a surprising representation theorem: every "context-free monoidal language" is the image under a "monoidal functor" of a "regular monoidal language".

\begin{remark}
  There is a minor difference between the definition of "regular monoidal languages" above, and those originally defined by \Sobocinski{} and the first author \cite{earnshaw22}, which were restricted to languages of scalar morphisms. In later work \cite{earnshaw_et_al:LIPIcs.MFCS.2023.43}, initial and final state finite languages were introduced. Bossut, Dauchet and Warin \cite{bossut} allow these to be regular languages; however, our applications require no more than a single initial and final word.
\end{remark}

\begin{remark}
  As defined, "regular monoidal languages" are subsets of \emph{free} strict "monoidal categories": we shall need only this case in order to prove our main theorem. "Context-free monoidal languages" will be defined over \emph{arbitrary} "strict monoidal categories", so this raises the question of extending the regular case to "monoidal categories" that are not free. We suggest this can be done by a generalization of \MZ{}'s definition of \emph{finite-state automata over a category} as finitary \emph{unique lifting of factorizations} functors \cite[Section 2]{mellies2}.
\end{remark}

\section{Context-Free Monoidal Languages} \label{sec:cfml}
We now turn our attention to context-free grammars over monoidal categories. The "multicategory of spliced arrows" is defined for any category. However, for categories equipped with a monoidal structure, it is natural to consider more general kinds of holes than allowed by the "spliced arrows" construction (\Cref{fig:monoidal-hole}). Rather than tuples of disjoint pieces, we should allow the possibility that a hole can be surrounded by strings. The necessity of considering these more general holes is forced upon us by various examples that could not be captured using "spliced arrows" (e.g. \Cref{ex:cftree,ex:hyperedge}). Proofs omitted from this section may be found in \Cref{ax:cfml}.

\subsection{The symmetric multicategory of diagram contexts}

"Context-free monoidal grammars" should contain productions from a variable to an incomplete diagram containing multiple variables or ``holes''.
This section constructs "diagram contexts" over an arbitrary "polygraph". "Diagram contexts" represent the incomplete derivation of a monoidal term: as such, they consist of string diagrams over which we add ``holes''. We shall notated these holes in string diagrams as pink boxes (e.g. \Cref{fig:monoidal-hole}).

\begin{figure}[h]
  \centering
  \includegraphics[scale=0.8]{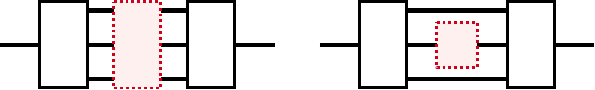}
  \caption{(Left) A "spliced arrow" is a tuple of morphisms. (Right) In a monoidal category, there is the possibility of more general holes, which do not split a morphism into disjoint pieces.}
  \label{fig:monoidal-hole}
\end{figure}

Substituting another diagram context inside a hole induces a symmetric multicategorical structure on the diagrams: symmetry means that we do not distinguish the specific order in which the holes appear. This allows us to avoid declaring a particular ordering of holes when defining a "context-free monoidal grammar".
The mathematical device that allows us to avoid declaring a particular ordering is \emph{shufflings}; their use in categorical logic is inspired by the work of Shulman \cite{shulman2016categorical}. 

\begin{definition}
  A shuffling of two lists, $Ψ ∈ \mathrm{Shuf}(Γ,Δ)$ is any list $Ψ$ that contains the elements of both $Γ$ and $Δ$ in any order but preserving the relative orders of $Γ$ and $Δ$.
\end{definition}

For instance, if $Γ = [\Ctx{x\vphantom{y}},\Ctx{y},\Ctx{z\vphantom{y}}]$ and $Δ = [\Ctx{u},\Ctx{v}]$, a shuffling is $Ψ = [\Ctx{x\vphantom{y}},\Ctx{u\vphantom{y}},\Ctx{y},\Ctx{z\vphantom{y}},\Ctx{v\vphantom{y}}]$, but not $[\Ctx{y},\Ctx{u\vphantom{y}},\Ctx{z\vphantom{y}},\Ctx{x\vphantom{y}},\Ctx{v\vphantom{y}}]$. The theory of diagram contexts will introduce a shuffling every time it mixes two contexts: this way, if a term was derived by combining two contexts, we can always reorder these contexts however we want. For instance, the term $\Ctx{u},\Ctx{v} ⊢ \Ctx{u} ⨾ \Ctx{v}$ was derived from composing the axioms $\Ctx{u} ⊢ \Ctx{u}$ and $\Ctx{v} ⊢ \Ctx{v}$; by choosing a different shuffling, we can also derive the term $\Ctx{v},\Ctx{u} ⊢ \Ctx{u} ⨾ \Ctx{v}$. Let us now formally introduce the theory.

\begin{definition}
  The theory of ""diagram contexts"" $\Ctx{\mathcal{P}}$ over a "polygraph", $\mathcal{P}$, is described by the following logic. This logic contains objects ($A,B,C,... ∈ \mathcal{P}_{obj}^{\ast}$) that consist of lists of types of the polygraph, $X,Y,Z,... ∈ \mathcal{P}_{obj}$; it also contains contexts $(\Gamma,\Delta,\Psi,...)$ that consist of lists of pairs of objects. Apart from the single variables $(x,y,z,..)$ and the generators of the "polygraph" $(f,g,h,...)$; we consider fully formed terms $(t_1,t_2,...)$.
  \begin{gather*}
    \inferrule[Identity]
      {\quad}
      {⊢ \mathrm{id} : \biobj{X}{X}} \qquad
     \inferrule[Generator]
      {\quad}
      {⊢ f : \biobj{X₁,...,Xₙ}{Y₁,...,Yₘ}} \qquad
    \inferrule[Hole]
      {\qquad}
      {\Ctx{x} : \biobj{A}{B} ⊢ \Ctx{x} : \biobj{A}{B}} \\
    \inferrule[Sequential]
      {Γ ⊢ t₁ : \biobj{A}{B} \quad Δ ⊢ t₂ : \biobj{B}{C}}
      {\mathrm{Shuf}(Γ;Δ) ⊢ t₁ ⨾ t₂ : \biobj{A}{C}} \qquad
    \inferrule[Parallel]
      {Γ ⊢ t₁ : \biobj{A₁}{B₁} \quad Δ ⊢ t₂ : \biobj{A₂}{B₂}}
      {\mathrm{Shuf}(Γ;Δ) ⊢ t₁ ⊗ t₂ : \biobj{A₁ ++ A₂}{B₁ ++ B₂}}
  \end{gather*}
  Every term in a given context has a unique derivation. We consider terms up to $\alpha$-equivalence and we impose the following equations over the terms whenever they are constructed over the same context: $(t₁ ⨾ t₂) ⨾ t₃ = t₁ ⨾ (t₂ ⨾ t₃); \quad t ⨾ \mathrm{id} = t; \quad t₁ \smallotimes (t₂ \smallotimes t₃) = (t₁ \smallotimes t₂) \smallotimes t₃; \quad (t₁ ⨾ t₂) \smallotimes (t₃ ⨾ t_4) = (t₁ \smallotimes t₃) ⨾ (t₂ \smallotimes t_4).$
\end{definition}

\begin{restatable}[]{proposition}{holessym}
  The multicategory of derivable sequents in the theory of "diagram contexts" is symmetric. In logical terms, exchange is admissible in the theory of "diagram contexts".
\end{restatable}

\begin{restatable}[]{proposition}{holesmonoidal}
  Derivable sequents in the theory of "diagram contexts" over a polygraph $𝓟$ form the free strict monoidal category over the polygraph extended with special ``hole'' generators,
  $𝓟 + \{ h_{A,B} : A → B \mid A , B ∈ 𝓟_{obj}^{\ast} \}$. Derivable sequents over the empty context form the free "strict monoidal category" over the polygraph $𝓟$. 
  Moreover, there exists a symmetric multifunctor $i ፡ \Ctx{\U{\fpro{𝓟}}} → \Ctx{\vphantom{\U{}}𝓟}$ interpreting each monoidal term as its derivable sequent.
\end{restatable}

\begin{remark}
  Various notions of ``holes in a monoidal category'' exist in the literature, under names such as \emph{optics}, \emph{contexts}, or \emph{wiring diagrams} \cite{Patterson_2021,roman21}. 
  Hefford and the authors \cite{roman23, produoidal23} gave a universal characterization of the \emph{produoidal} category of optics over a "monoidal category". This produoidal structure is useful for describing \emph{decompositions} of diagrams. The above logic generates a multicategory similar to the operad of directed, acyclic \emph{wiring diagrams} introduced by Patterson, Spivak and Vagner \cite{Patterson_2021}; whose operations are generic morphism shapes, rather than holes in a specific monoidal category.
\end{remark}

\begin{definition} \label{defn:symmulti}
  A ""symmetric multigraph"" is a "multigraph" $G$ equipped with bijections $σ^{\ast} ፡ G(X_1, ..., X_n; Y) \cong G(X_{σ(1)}, ..., X_{σ(n)}; Y)$ for every list $X_1,...,X_n$ of sorts and every permutation $σ$, satisfying $(σ · τ)^{\ast} = σ^{\ast} ⨾ τ^{\ast}$ and $\id{}^{\ast} = \id{}$. A ""morphism of symmetric multigraphs"" is a "morphism of multigraphs" which commutes with the bijections.
\end{definition}
\begin{definition}
  Every "multigraph", $M$, freely induces a "symmetric multigraph", $\mathsf{clique}(M)$, with the same objects and, for each $f ∈ M(X₁,...,Xₙ;Y)$, a clique of elements 
  $$f_{σ} ∈ \mathsf{clique}(M)(X_{σ(1)}, ..., X_{σ(n)}; Y),$$ connected by symmetries, meaning that $σ^{\ast}(f_{\tau}) = f_{σ \cdot τ}$. This is the left adjoint to the inclusion of "symmetric multigraphs" into "multigraphs".
\end{definition}  
\begin{remark}
  Given any "symmetric multigraph" $G$, finding a "multigraph" $M$ whose clique recovers it, $\mathsf{clique}(M) = G$, amounts to choosing a representative for each one of the cliques of the multigraph. Any "symmetric multigraph" can be (non-uniquely) recovered in this way: for each multimorphism $f ∈ G(X₁,...,Xₙ;Y)$, we can consider its orbit under the action of the symmetric group, $\mathsf{orb}(f) = \{σ^{\ast}(f) \mid \sigma \in S_n\}$ -- the orbits of different elements may coincide, but each element does have one -- and picking an element $g_o$ for each orbit, $o ∈ \{\mathsf{orb}(f) \mid f ∈ G\}$, recovers a multigraph giving rise to the original symmetric multigraph.
\end{remark}

\begin{definition}
  The theory of "diagram contexts" over a finitely presented "monoidal category", $(𝓟,𝓔,l,r)$ (\Cref{defn:finpres}), is the theory of "diagram contexts" over its generators, quotiented by its equations; in other words, it is the equalizer of the two projections of each equation, interpreted as derivable sequents $(\Ctx{l^{\ast}}⨾i),(\Ctx{r^{\ast}}⨾i) ፡ \Ctx{𝓔} → \Ctx{𝓟}$.
\end{definition}

\begin{proposition} \label{prop:hole-functor}
  The formation of "diagram contexts" in a "monoidal category" or "polygraph" extends to functors $\Ctx{\phantom{ℂ}} : \MonCat \to \MultiCat$ and $\Ctx{\phantom{ℂ}} : \PolyGraph \to \MultiGraph$, which moreover commute with the free multicategory $\fmult{}$ and free monoidal category functors $\fpro{}$.
\end{proposition}

At this point, the reader may doubt that the formation of "diagram contexts" has a left adjoint similar to the "contour" functor for "spliced arrows". Indeed, in order to recover a left adjoint, we shall need to introduce another multicategory of diagrams which we call "raw optics". This technical device will allow us to prove our main theorem (\Cref{thm:main}). However, let us first see the definition of context-free monoidal grammar, and some examples.

\subsection{Context-Free Monoidal Grammars}
We now have the ingredients for our central definition. A "context-free monoidal grammar" specifies a language of string diagrams by a collection of rewrites between \emph{diagram contexts}, where the non-terminals of a context-free grammar are now (labelled) \emph{holes} in a diagram (e.g. Figure \ref{fig:cfmg-example}). Our definition is entirely analogous to \Cref{defn:mzcfg}, but using our new symmetric multicategory of "diagram contexts" in a monoidal category, instead of "spliced arrows".

\begin{definition} \label{defn:cfmg}
  A ""context-free monoidal grammar"" over a strict monoidal category $(ℂ,\otimes,I)$ is a morphism of "symmetric multigraphs" $\Psi : \mathcal{G} \to |\Ctx{ℂ}|$, into the underlying multigraph of "diagram contexts" in $ℂ$, where $\mathcal{G}$ is finite, and a start sort $S_{X,Y} \in \Psi^{-1}(\biobj{X}{Y})$. %
\end{definition}

We shall use the notation $S \sqsubset \biobj{A}{B}$ to indicate that $\Psi(S) = \biobj{A}{B}$, following the convention in the literature \cite{mellies2}.

A "morphism of symmetric multigraphs" $\Psi : \mathcal{G} \to |\Ctx{ℂ}|$ defining a grammar uniquely determines, via the free-forgetful adjunction, a symmetric multifunctor $\hat{\Psi} : \fmult{\mathcal{G}} \to \Ctx{ℂ}$, mapping (closed) derivations to morphisms of $ℂ$. The language of a grammar is then defined analogously to \Cref{defn:mzcfglang}:

\begin{definition} \label{def:cfml}
  Let $(\Psi : \mathcal{G} \to |\Ctx{ℂ}|, S \sqsubset \biobj{A}{B})$ be a "context-free monoidal grammar". The language of $\Psi$ is the set of morphisms in $\mathbb{C}(A;B)$ given by the image under $\hat{\Psi}$ of the set of derivations $\fmult{\mathcal{G}}(;S)$. A set of morphisms $L$ in $ℂ$ is a ""context-free monoidal language"" if and only if there exists a "context-free monoidal grammar" whose language is $L$. %
\end{definition}

Let us see some examples.

\begin{example}[Classical context-free languages]
  Every "context-free monoidal grammar" of the following form is equivalent to a classical context-free grammar of words. Let $Γ$ be a single-sorted finite "polygraph" whose generators are all of arity and coarity 1. Then "context-free monoidal grammars" over $\fpro{Γ}$ with a start symbol $\phi(S) \sqsubset \biobj{1}{1}$ are context-free grammars of words over $Γ$. \Cref{fig:brackets} gives the classical example of balanced parentheses. Similarly, every context-free grammar of words may be encoded as a "context-free monoidal grammar" in this way.
\end{example}

\begin{figure}[h]
  \centering
  \includegraphics[width=\textwidth]{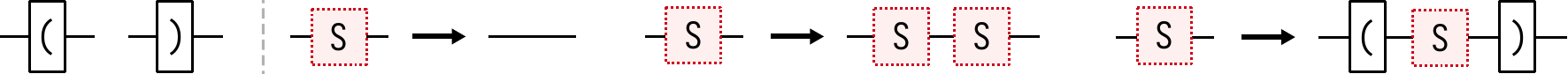}
  \caption{Balanced parentheses as a context-free monoidal grammar over the polygraph $Γ$ (left).}
  \label{fig:brackets}
\end{figure}

\begin{example}[Context-free tree grammars] \label{ex:cftree}
  \emph{Context-free tree grammars} \cite{10.5555/267871.267872,10.1145/800169.805428} are defined over \emph{ranked alphabets} of terminals and non-terminals, which amount to "polygraphs" in which the "generators" have arbitrary arity (the rank) and coarity $1$. Productions have the form $A(x_1, ..., x_m) \to t$ where the left hand side is a non-terminal of rank $m$ whose frontier is labelled by the variables $x_i$ in order, and whose right hand side is a tree $t$ built from terminals and non-terminals, and whose frontier is labelled by variables from the set $\{x_1, ..., x_m\}$. Note that $t$ may use the variables non-linearly.

  For example, let $S$ be a non-terminal with coarity $0$, $A$ a non-terminal with coarity $2$, $f$ a terminal of coarity $2$, and $x$ a terminal of coarity $0$ (a leaf). Then a possible rule over these generators is $A(x_1, x_2) \to f(x_1 , A(x_1,x_2))$, where $x_1$ appears non-linearly.
  In order to allow such non-linear use of variables in a "context-free monoidal grammar", we can consider the free \emph{cartesian} category over $Γ$.
  In terms of string diagrams, this amounts to introducing new generators for copying (\scalebox{0.7}{\raisebox{.5em}{\begin{tikzpicture}[scale=.75,baseline=(current bounding box.center)]
\comult
\end{tikzpicture}
}}) and deleting variables (\scalebox{0.7}{\raisebox{.25em}{\begin{tikzpicture}[xscale=.5,baseline=(current bounding box.center)]
\counit
\end{tikzpicture}}}), satisfying some equations which we recall in \Cref{ax:cartesian}.

  Let $Γ$ be a "polygraph" in which "generators" have arbitrary arity, and coarity $1$, as above.
  "Context-free monoidal grammars" over the free "cartesian category" on $Γ$, with a start symbol $S \sqsubset \biobj{0}{1}$ are context-free tree grammars.
  In \Cref{fig:tree} we extend the above data to a full example.
  Note that by allowing start symbols $S \sqsubset \biobj{0}{n}$, we can produce forests of $n$ trees.
\end{example}

\begin{figure}[h]
  \centering\includegraphics[width=0.8\columnwidth]{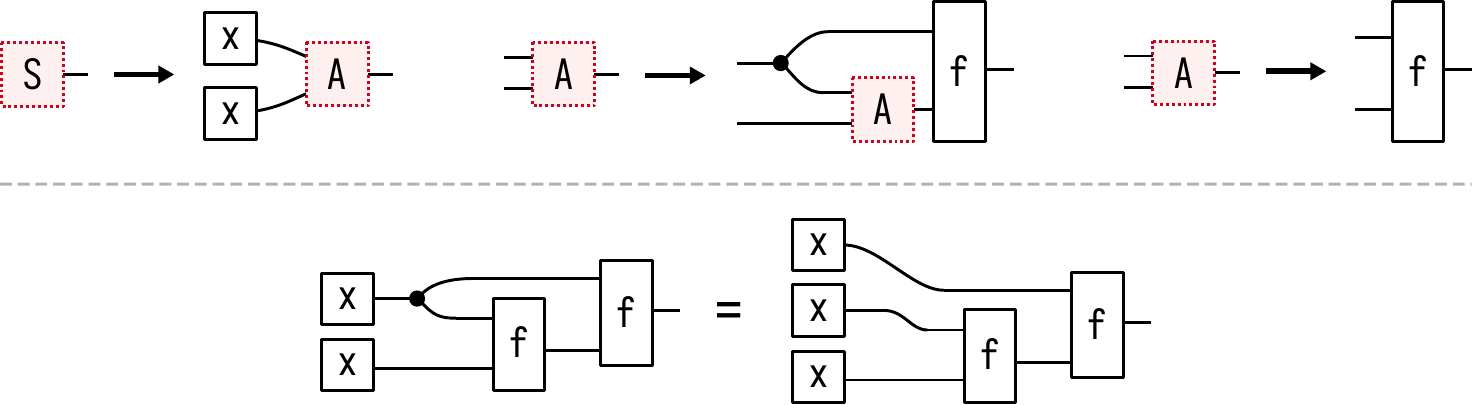}
  \caption{Example of a context-free tree grammar as a context-free monoidal grammar. The string diagrams at the bottom are \emph{equal} in the free \emph{cartesian} category over the polygraph of terminals.}
  \label{fig:tree}
\end{figure}

\begin{example}[Regular monoidal languages] \label{prop:reg-are-cf}
   "Regular monoidal languages" are "context-free monoidal languages": let $R = (\phi : G \to Γ, i, f)$ be a "regular monoidal grammar". We lift this to a "context-free monoidal grammar" using the diagram contexts functor (\Cref{prop:hole-functor}), giving $\Ctx{\phi} : \Ctx{G} \to \Ctx{Γ}$, taking the free monoidal functor $\fpro{\Ctx{\phi}}$ on this morphism of polygraphs, and finally commuting $\fpro{}$ and $\Ctx{\phantom{ℂ}}$ (\Cref{prop:hole-functor}). The pair $\biobj{i}{f}$ provides the start symbol.
\end{example}

\begin{example}[Hyperedge-replacement grammars] \label{ex:hyperedge}
  Hyperedge-replacement (HR) grammars are a kind of \emph{context-free graph grammar} \cite{Engelfriet1997}. We consider HR grammars in \emph{normal form} in the sense of Habel \cite[Theorem 4.1]{habel}. A production $N \to R$ of an HR grammar has $N$ a non-terminal with arity and coarity, and $R$ a hypergraph with the same arity and coarity (a \emph{multi-pointed hypergraph} in Habel's terminology), whose hyperedges are labelled by some finite set of terminals and non-terminals.

  Just as trees are morphisms in free \emph{cartesian} monoidal categories (\Cref{ex:cftree}), hypergraphs are the morphisms of monoidal categories equipped with extra structure, known as \emph{hypergraph categories} \cite{rosebrugh2005generic,10.1145/3502719,FONG20194746}. "Generators" in a "polygraph" are exactly directed hyperedges. The extra structure in a hypergraph category, which we recall in \Cref{ax:hypergraph}, amounts to a combinatorial encoding of patterns of wiring between nodes.

  Let $Γ$ be a "polygraph" of terminal hyperedges, $G$ a "multigraph" of non-terminal rules, and $S \in G$ a start symbol. Then "context-free monoidal grammars" $(G \to \U{\textsf{Hyp}[Γ]}, S)$ over the "free hypergraph category" on $Γ$ are exactly hyperedge replacement grammars over $Γ$ (e.g. \Cref{fig:hypergraph}). A hole in a morphism in $\textsf{Hyp}[Γ]$ is a placeholder for an $(n,m)$ hyperedge, the grammar labels these holes by non-terminals, and composition corresponds to hyperedge replacement.
\end{example}

\begin{figure}[h]
  \centering\includegraphics[width=\columnwidth]{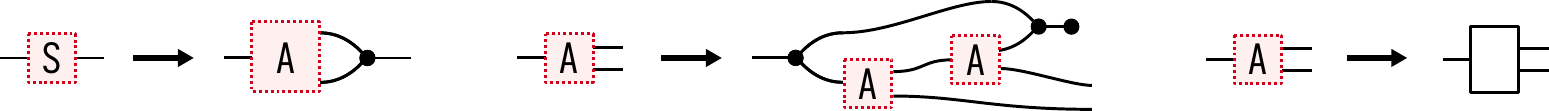}
  \caption{A hypergraph grammar for simple \emph{control flow graphs} with branching and looping, as a context-free monoidal grammar. Based on Habel \cite[Example 3.3]{habel}.}
  \label{fig:hypergraph}
\end{figure}

\begin{example}[Unbraids] \label{ex:unbraids}
  We return to the language of unbraids suggested in \Cref{prop:unbraid-reg}. Take the grammar over the over- and under-braiding "polygraph" depicted in \Cref{fig:cfmg-example}, with start symbol $S \sqsubset \biobj{2}{2}$. The language of this grammar consists of unbraids on two strings. %
\end{example}

\begin{figure}[h]
  \centering\includegraphics[width=\textwidth]{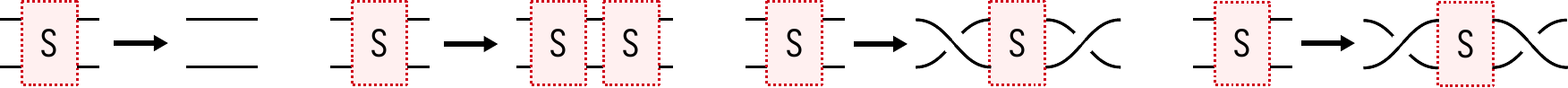}
  \caption{A context-free monoidal grammar of unbraids, with start symbol $S$.}
  \label{fig:cfmg-example}
\end{figure}

Let us record some basic closure properties of context-free monoidal languages.

\begin{proposition}
  Context-free monoidal languages over $ℂ$ with start symbol $S \sqsubset \biobj{X}{Y}$ are closed under union. The underlying morphism is given by the copairing, and start symbols can be unified by introducing a fresh symbol and productions where necessary, as in the classical case. Context-free monoidal languages are also closed under images of "strict monoidal functors": the underlying morphism is given by postcomposition. %
\end{proposition}

\section{Optical Contour of a Multicategory} \label{sec:contour}
An important realization of \MZ{} is that the formation of "spliced arrows" in a category has a left adjoint, which they call the \emph{contour} of a multicategory \cite[Section 3.2]{mellies2}. This adjunction is a key conceptual tool in their generalized version of the \CS{} representation theorem, and is closely linked to the notion of \emph{item} in LR parsing \cite{mellies2}. In this section, we present a similar adjunction for the monoidal setting. However, it is not clear that the formation of "diagram contexts" has a left adjoint. We must therefore first conduct a dissection of "diagram contexts" into \emph{raw optics}.

\subsection{The multicategory of raw optics}
A "raw optic" is a tuple of morphisms obtained by cutting a "diagram context" into a sequence of disjoint pieces (cf. \cite{prof2001-07488}). In \Cref{sec:contour} we shall see that "raw optics" has a left adjoint, the "optical contour", and this will be enough to prove our representation theorem (\Cref{thm:main}).

\begin{definition}\label{def:rawoptics}
  The ""multicategory of raw optics"" over a "strict monoidal category" $ℂ$, denoted $\raw[\mathbb{C}]$, is defined to have, as objects, pairs $\biobj{A}{B}$ of objects of $ℂ$, and, as its set of multimorphisms, $\raw[\mathbb{C}](\biobj{A_1}{B_1}, ..., \biobj{A_n}{B_n}; \biobj{S}{T})$, the following set
  \[\begin{aligned}
    \sum_{\mathclap{M_i, N_i \in ℂ}} \mathbb{C}(S; M_1 \!\smallotimes\! A_1 \!\smallotimes\! N_1) ×
    \prod_{i=1}^{n-1} \mathbb{C}(M_i \!\smallotimes\! B_i \!\smallotimes\!  N_i; M_{i+1} \!\smallotimes\! A_{i+1} \!\smallotimes\! N_{i+1})
× \mathbb{C}(M_n \!\smallotimes\! B_n \!\smallotimes\! N_n; T).
  \end{aligned}\]
  As a special case, $\raw[\mathbb{C}](;\biobj{S}{T}) := \mathbb{C}(S;T)$. 
  In other words, a multimorphism, from $\biobj{A_1}{B_1}, ..., \biobj{A_n}{B_n}$ to $\biobj{S}{T}$, consists of two families of objects, $M₁,…,Mₙ$ and $N₁,...,Nₙ$, and a family of functions, $(f_0, ..., f_n)$, with types $f₀ ፡ S → M_1 \smallotimes A_1 \smallotimes N_1$; with $fᵢ ፡ M_i \smallotimes B_i \smallotimes N_i \to M_{i+1} \smallotimes A_{i+1} \smallotimes N_{i+1};$ for each $1 \leq i \leq n-1$; and $fₙ ፡ M_n \smallotimes B_n \smallotimes N_n \to T$. In the special nullary case, we have a single morphism $f₀ ፡ S → T$.

  Identities are given by pairs $(\id{A}, \id{B})$. Given two "raw optics" $f = (f_0, ..., f_n)$ and $g = (g_0, ..., g_m)$, their composition is defined by 
$$f \comp_i g := (g_0, ..., g_i \comp (\id{} ⊗ f_0 ⊗ \id{}), ..., \id{} ⊗ f_i ⊗ \id{}, ..., (\id{} ⊗ f_n ⊗ \id{}) \comp g_{i+1}, ..., g_n).$$
\end{definition}

\begin{figure}
  \includegraphics[width=\textwidth]{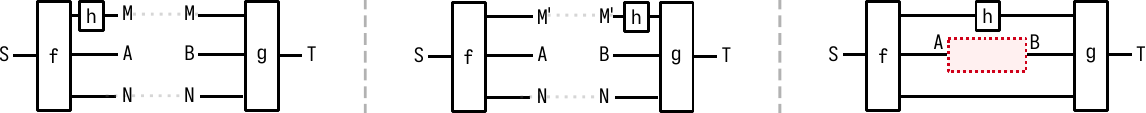}
  \caption{Two "raw optics" (left, centre) in $\raw[ℂ](\biobj{A}{B};\biobj{S}{T})$ which quotient to the same "diagram context". Note that a "raw optic" is not the same as a "spliced arrow": the types $M,N$ must match.}
  \label{fig:raw-optics}
\end{figure}

Every "raw optic" can be glued into a "diagram context", as illustrated in \Cref{fig:raw-optics}. More precisely we have,

\begin{proposition} \label{prop:raw-to-diag}
  There is an identity on objects multifunctor $q : \raw[ℂ] \to \Ctx{ℂ}$ mapping each "raw optic" to its corresponding diagram context. Equivalently, there is an identity on objects symmetric multifunctor $q^{\ast} ፡ \mathsf{clique}(\raw[ℂ]) → \Ctx{ℂ}$; this symmetric multifunctor is full.
\end{proposition}

\begin{proposition}
  The construction of "raw optics" extends to a functor
  $\raw : \MonCat → \MultiCat$
  between the categories of "strict monoidal categories" and "strict monoidal functors", and "multicategories" and "multifunctors".
\end{proposition}

\begin{remark}
  We could have defined "context-free monoidal grammars" as morphisms into "raw optics", rather than "diagram contexts", but this would require an arbitrary choice of "raw optic" for each rule, as in \Cref{fig:raw-optics}. In particular, this would force us to choose a particular ordering of the holes, since "raw optics" do not form a \emph{symmetric} multicategory. On the other hand, that such a choice exists will be needed to prove our representation theorem (\Cref{sec:monoidal-cs}).
\end{remark}

\subsection{Optical contour}
We now introduce the left adjoint to the formation of "raw optics", which we call the \emph{optical contour} of a multicategory. The difference from the contour recalled in \Cref{sec:prelims} is that additional objects $M_i, N_i$ are introduced which keep track of strings that might surround holes, and this gives rise to a "strict monoidal category".%

\begin{figure}[h]
  \centering
  \includegraphics[scale=0.5]{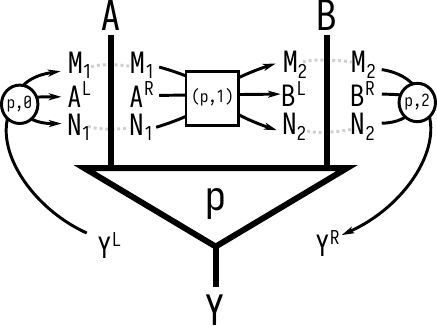}
  \caption{A multimorphism $p \in 𝕄(A,B;Y)$ and its three "sectors" given by "optical contour": $(p,0) : Y^L \to M_1⊗A^L⊗N_1, (p,1) : M_1⊗A^R⊗N_1 \to M_2⊗B^L⊗N_2, (p,2) : M_2⊗B^R⊗N_2 \to Y^R$.}
  \label{fig:optical-contour}
\end{figure}

\begin{definition}\label{def:opticalContour}
  Let $\mathbb{M}$ be a multicategory. Its ""optical contour"", $𝒞\mathbb{M}$, is the "strict monoidal category" presented by a "polygraph" whose generators are given by \emph{contouring} multimorphisms in $\mathbb{M}$. Each multimorphism gives rise to a set of generators for the monoidal category $𝒞\mathbb{M}$ -- its set of ""sectors"", as in \Cref{fig:optical-contour}.

  Explicitly, for each object $A \in \mathbb{M}$, the optical contour $𝒞\mathbb{M}$ contains a left polarized, $A^L$, and a right polarized, $A^R$, version of the object. Additionally, for each multimorphism $f \in \mathbb{M}(X_1, ..., X_n; Y)$, there exists a family of objects $M_1^f, ..., M_n^f, N_1^f, ..., N_n^f$, whose superscripts we omit when they are clear from context. The morphisms are given by the following generators. For each $f \in \mathbb{M}(X_1, ..., X_n; Y)$, we consider the following $n+1$ generators:
  \[\begin{aligned}
    & (f,0) : Y^{L} \to M_1^f \smallotimes X_1^L \smallotimes N_1^f, \\
    & (f,i) : M_i^f ⊗ X_i^R ⊗ N_i^f \to M_{i+1}^f ⊗ X_{i+1}^L ⊗ N_{i+1}^f, 
       \mbox{ for }1 \leqslant i\leqslant n-1,\mbox{ and } \\
    & (f,n) : M_n^f ⊗ X_n^R ⊗ N_n^f \to Y^R.
  \end{aligned}\]
  In particular, for a nullary multimorphism $f \in \mathbb{M}(;Y)$, we consider a generator $(f,0) : Y^L \to Y^R$. Further, we ask for the following equations which ensure that the optical contour preserves identities and composition: for all $x \in 𝕄$, $(\id{X}, 0) = \id{X^L}, (\id{X}, 1) = \id{X^R}$ with $M_1^{\id{X}} = N_1^{\id{X}} = I$; and given any $f \in 𝕄(X_1, ..., X_n; Y_i)$ and $g \in 𝕄(Y_1, ..., Y_m; Z)$,
   \[
    (f \comp_i g, j) = \begin{cases}
      (g,j) & j<i, \text{ with } M_j^{f\comp g}= M_j^g, N_j^{f \comp g} = N_j^g \\
      (g,i)\comp(\id{} ⊗ (f,0) ⊗ \id{}) & j = i, \text{ with } M_i^{f\comp g}= M_j^g ⊗ M_0^f \\
      \id{M_i^g}⊗(f,j-i)⊗\id{N_i^g} & i<j<i+n, \text{ with } M_j^{f\comp g}= M_i^g⊗M_{j-i}^f \\
      (\id{} ⊗ (f,n) ⊗ \id{})\comp(g,i+1) & j = i+n+1,  \text{ with } M_j^{f\comp g}= M_i^g⊗M_n^f \\
      (g,j-n) & j > i+n+1 \text{ with } M_j^{f\comp g}= M_{j-n}^g.
    \end{cases}
  \]
  In particular, when $f \in 𝕄(;Y_i)$ is nullary, $(f \comp_i g, 0) = g_i \comp f_0 \comp g_{i+1}$.
\end{definition}

\begin{restatable}[]{theorem}{opticadjoint} \label{thm:opticadjoint}
  "Optical contour" is left adjoint to "raw optics"; there exists an adjunction $(\mathcal{C} \dashv \raw{}) : "\MonCat" \to \MultiCat.$
\end{restatable}
\begin{proof}%
  See \Cref{ax:sec:opticadjoint}.
  \vspace*{-3mm}
\end{proof}

\section{A Monoidal Representation Theorem} \label{sec:monoidal-cs}

The Chomsky-Schützenberger representation theorem says that every context-free language can be obtained as the image under a homomorphism of the intersection of a Dyck language and a regular language \cite{chomsky63}. \MZ{} \cite{mellies2022parsing} use their splicing-contour adjunction to give a novel proof of this theorem for context-free languages in categories: the classial version is recovered when the category is a free monoid. The role of the Dyck language, providing linearizations of derivation trees, is taken over by \emph{contours} of "derivations".

Monoidal categories provide a more striking case: the Dyck language is not needed because the information that parentheses encode can be carried instead by tensor products. In this section, we show that a regular monoidal language of "optical contours" is sufficient to reconstruct the original language. \Cref{thm:main} states that \emph{every "context-free monoidal language" is the image under a "monoidal functor" of a "regular monoidal language"}.

Our strategy will be to first choose a factoring of a grammar into raw optics, then use the "optical contour"/"raw optics" adjunction to produce the required monoidal functor. We must first establish that such a factoring exists. Omitted proofs may be found in \Cref{ax:monoidal-cs-proofs}.

\begin{restatable}[]{lemma}{factor} \label{lem:factor}
  Any "morphism of symmetric multigraphs" underlying a "context-free monoidal grammar", $\phi : G \to \U{\Ctx{ℂ}}$, factors (non-uniquely) through the quotienting of "raw optics" (\Cref{prop:raw-to-diag}); meaning that there exists some "multigraph" $G'$ satisfying $G = \mathsf{clique}(G')$, and some morphism $\rawr{\phi} \colon G' → \U{\raw[ℂ]}$, such that  $\phi = \mathsf{clique}(\rawr{\phi}) ⨾ q^{\ast}$.
\end{restatable}

Call the factor $\rawr{\phi} : G' \to \U{\raw[ℂ]}$ a ""raw representative"" of $\phi$. It amounts to choosing a fixed ordering of the holes in a "diagram context" for each rule in the grammar, and a particular splicing into a "raw optic".

\begin{restatable}[]{lemma}{rawlang} \label{lem:raw-lang}
  Let $\mathcal{G} = (\phi, S)$ be a "context-free monoidal grammar". Then the language of any "raw representative" $\rawr{\phi}$ of $\phi$ (with start symbol $S$) equals the language of $\mathcal{G}$. That is, $\rawr{\phi}[\fmult{G'}(;S)]  = \phi[\fmult{G}(;S)]$.
\end{restatable}

\begin{restatable}[]{lemma}{uniquemonoidal} \label{lem:uniquemonoidal}
  A "raw representative" $\rawr{\phi} : G' \to \U{\raw[ℂ]}$ uniquely determines a "strict monoidal functor" $I_\phi : \fpro{(\cont{G'})} \to \mathbb{C}$.
\end{restatable}

We shall see that this "monoidal functor" maps the following "regular monoidal language" over $\cont{G}$ to the language of the original "context-free monoidal grammar" $\phi$.

\begin{definition} \label{defn:regular-representation}
  Let $\mathcal{G} = (\phi : G \to \U{\Ctx{ℂ}}, S)$ be a "context-free monoidal grammar", and $\rawr{\phi}$ a "raw representative" with domain $G'$. Define a ""regular representative"" of $\mathcal{G}$ to be the "regular monoidal grammar" $\mathcal{R} =  (\id{} ፡ \cont{G'} → \cont{G'}, S^L, S^R)$ over "optical contours" of $G'$ whose "morphism of polygraphs" is the identity. %
\end{definition}

\begin{restatable}[]{lemma}{bij} \label{lem:bij}
  Given a multigraph $G$, there is a bijection between "derivations" rooted at a sort $S$ and "optical contours" from $S^L$ to $S^R$, that is $\fmult{G}(;S) \cong \fpro{(\mathcal{C}G)}(S^L; S^R)$. %
\end{restatable}

\begin{restatable}[]{theorem}{main} \label{thm:main}
  The language of a "context-free monoidal grammar" $\mathcal{G} = (\phi : G \to \U{\Ctx{ℂ}},S)$ equals the image of a "regular representative" under the "monoidal functor" $I_\phi$ of \Cref{lem:uniquemonoidal}.
\end{restatable}

\Cref{thm:main} is at first quite surprising, since in comparison with the usual \CS{} theorem and its generalization \cite{mellies2}, one might expect to see an \emph{intersection} of a "regular monoidal language" and a "context-free monoidal language". Instead, this theorem tells us that "regular monoidal languages" are powerful enough to encode "context-free monoidal languages", even while the latter is strictly more expressive than the former. Just as a context-free grammar suffices to specify a programming language which may encode instructions for arbitrary computations, "regular monoidal languages" can specify arbitrary "context-free monoidal languages", with a "monoidal functor" effecting the ``compilation''.

\section{Conclusion}
There are still many avenues to explore in this structural approach to context-free languages. One obvious direction is to investigate a notion of pushdown automaton for context-free monoidal languages. In fact, it still remains to be elaborated how pushdown automata emerge for context-free languages over plain categories. Following the general principle of \emph{parsing as a lifting problem} \cite{mellies2}, and the duality of grammars (fibered) and automata (indexed) may provide some clue to characterizing such automata by a universal property.

The study of languages and the dependence relations that "diagram contexts" naturally present may be useful to the study of complexity in monoidal categories, such as the notion of ``monoidal width'' proposed by Di Lavore and \Sobocinski{} \cite{dilavore,DiLavoreS23}. Conversely, measures of monoidal complexity may inform the cost of parsing different terms.

Finally, different types of string diagram exist for a variety of widely applied categorical structures beyond monoidal categories, such as double categories \cite{myers2018string}. There are many opportunities to extend the general principle elaborated here to a notion of context-free language in these structures.

\bibliography{bibliography.bib}

\newpage
\appendix
\section{Monoidal Categories} \label{ax:monoidals}
\begin{definition}
  A ""strict monoidal category"" $ℂ$ consists of a monoid of objects,
  or resources, $(ℂ_{obj}, ⊗, I)$, and a collection of morphisms, or processes, $ℂ(X; Y)$,
  indexed by an input $X ∈ ℂ_{obj}$ and an output $Y ∈ ℂ_{obj}$. A strict monoidal category
  is endowed with operations for the sequential and parallel composition of processes,
  respectively
  \begin{gather*}
    (⨾) ፡ ℂ(X; Y ) × ℂ(Y ; Z) → ℂ(X; Z), \\
    (⊗) ፡ ℂ(X; Y ) × ℂ(X'; Y') → ℂ(X ⊗ X'; Y ⊗ Y'),
  \end{gather*}
  and a family of identity morphisms, $id_{X} ∈ C(X; X)$. Strict monoidal categories
  must satisfy the following axioms.
  \begin{enumerate}
    \item Sequencing is unital, $f ⨾ \id{Y} = f$ and $\id{X} ⨾ f = f$.
    \item Sequencing is associative, $f ⨾ (g ⨾ h) = (f ⨾ g) ⨾ h$.
    \item Tensoring is unital, $f ⊗ \id{I} = f$ and $\id{I} ⊗ f = f$.
    \item Tensoring is associative, $f ⊗ (g ⊗ h) = (f ⊗ g) ⊗ h$.
    \item Tensoring and identities interchange, $\id{A} ⊗ \id{B} = \id{A⊗B}$.
    \item Tensoring and sequencing interchange,
    $$(f ⨾ g) ⊗ (f' ⨾ g') = (f ⊗ f') ⨾ (g ⊗ g').$$
  \end{enumerate}
\end{definition}

\begin{remark}
  This definition, slightly different from that found in most references, is taken from the thesis of Román \cite{roman23}. %
\end{remark}

\begin{definition} \label{defn:symmoncat}
  A ""symmetric strict monoidal category"" is a ""monoidal category"" equipped with a natural family of isomorphisms $\sigma_{X,Y} : X⊗Y \to Y⊗X$ for every pair of objects $X,Y$. We can extend string diagrams to express this structure, by allowing strings to cross without tangling. That is, we introduce components (below, left) for every pair of sorts, and equations (below, right) expressing that these are natural isomorphisms. Adding this structure to the free monoidal category over a polygraph presents the free symmetric monoidal category over that polygraph.

  \begin{figure}[h]
    \centering
    \includegraphics[width=\textwidth]{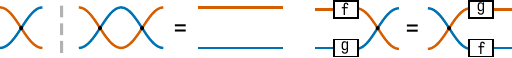}
    \end{figure}
\end{definition}

\begin{definition}
  A ""strict monoidal functor"", $F ፡ ℂ → 𝔻$, is a monoid homorphism between their objects, $F_{obj} ፡ ℂ_{obj} → 𝔻_{obj}$, and an assignment of morphisms $f ∈ ℂ(X;Y)$ to morphisms $F(f) ∈ 𝔻(FX; FY)$. A functor must preserve sequential composition, $F(f ⨾ g) = F(f) ⨾ F(g)$; parallel composition, $F(f ⊗ g) = F(f) ⊗ F(g)$; and identities, $F(\id{}) = \id{}$. Strict monoidal categories with strict monoidal functors form a category, ""\MonCat"". %
\end{definition}

\section{Pumping lemma for regular monoidal languages}
\begin{lemma}[\cite{earnshaw24}] \label{lem:pump}
  Let $L$ be a "regular monoidal language". Then $\forall k \in \mathbb{N}^{+}, \exists n$ such that for any $s \in L$ where $s$ may be factorized into $m \geqslant n$ non-identity morphisms $s = s_0\comp ... \comp s_i \comp ...\comp s_m$ where $s_i : k_i \to k_{i+1}$, with $1 \leqslant k_i \leqslant k$,
  there exists $i,j,\ell$ such that $k_i = k_j = \ell$ and $s' \comp (s'')^a \comp s''' \in L$ for all $a \geqslant 0$,
  where $s' = s_0 \comp ... \comp s_i$, $s'' = s_{i+1} \comp ... \comp s_j$, and $s''' = s_{j+1} \comp ... \comp s_m$
\end{lemma}
\begin{proof}
Let L be the language of a grammar ($\phi : M \to \Gamma, I, F)$. If $L$ has a finite number of connected string diagrams, then for any $k$ take $n$ be longer than the longest factorization over all diagrams in $L$, then the lemma holds vacuously. Otherwise let $k$ be given, then take $n = \sum_{i=0}^k |S_M|^i$. Let $s \in L$, such that it has a factorization of the form above. Then by the pigeonhole principle, we will have $i,j,\ell$ as required in the lemma.
\end{proof}

\begin{lemma}[Contrapositive form] \label{lem:pumpc}
Let L be a language and suppose that $\exists k \in \mathbb{N}^{+}$ such that $\forall n$ there exists a morphism $w \in L$ that factorizes as in \Cref{lem:pump} and for all $i,j,\ell$ such that $k_i = k_j = \ell$, there exists an $a$ such that the pumped morphism $w'w''^{a}w''' \notin L$, then L is not regular monoidal.
\end{lemma}

\begin{observation}
  This reduces to the pumping lemma for words and trees, taking $k=1$.
\end{observation}

\section{Optical contour-splice adjunction} \label{ax:sec:opticadjoint}

\opticadjoint*
\begin{proof}
  Let $ℂ$ be a strict monoidal category and let $𝕄$ be a multicategory. We need first to prove that the two constructions involved, $\cont$ and $\raw$, are indeed functors -- this proof, although tedious, proceeds as expected and we prefer to omit it here.
  
  We will show that there is a bijection between strict monoidal functors $\cont{𝕄} → ℂ$ and multifunctors $𝕄 → \raw[ℂ]$.
  \begin{itemize}
    \item The objects of $\raw[ℂ]$ are pairs of objects. Mapping an object of the multicategory $X ∈ 𝕄$ to a pair of objects is the same as mapping two objects, $X^L$ and $X^R$, to the objects of the category $ℂ$.
    \item Mapping a multimorphism $f ∈ 𝕄(X_1,...,X_n;Y)$ to the multicategory of raw optics consists of choosing a family of functions $(f₀,...,f_n)$ together with two families of objects $M₁,..,Mₙ$ and $N₁,...,Nₙ$. This is the same choice we need to map each one of the components of the contour of $f ∈ 𝕄(X_1,...,X_n;Y)$ to that exact family of functions.
  \end{itemize}
  That is, we have only checked that, by construction, the maps out of the contour correspond with multifunctors to raw optics. The adjunction remains conceptually interesting because it links two concepts that have different conceptual interpretations, even if it can be reduced to note that one has been defined as the adjoint to the other.
\end{proof}

\section{Cartesian monoidal categories} \label{ax:cartesian}
The free ""cartesian category"" over a polygraph may be presented using string diagrams. As with "symmetric monoidal categories", we add some new generators and equations, to the effect that every object is equipped with a natural and uniform counital comagma structure. That is, in addition to symmetric structure, we add the following generators and equations.

\begin{figure}[h]
  \centering
  \includegraphics[width=0.8\textwidth]{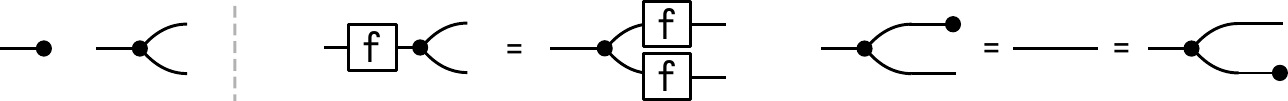}
\end{figure}

This structure must be uniform, in the sense that the structure on tensor products is given by tensor products of structure. See \cite[Section 4.1]{selinger2011} for more details.

\begin{remark}
  In most sources, cartesian categories are presented in terms of the presence of cocommutative comonoid structure. However, Román has shown that counital comagmas suffice \cite{roman23}.
\end{remark}

\section{Hypergraph categories} \label{ax:hypergraph}
The free ""hypergraph category"" over a polygraph may be presented using string diagrams. As with "cartesian monoidal categories", we add some new generators and equations. This extra structure amounts to equipping every object with the structure of a special commutative frobenius algebra. That is, in addition to symmetry we ask for the following generators and equations:

\begin{figure}[h]
  \centering
  \includegraphics[width=0.8\textwidth]{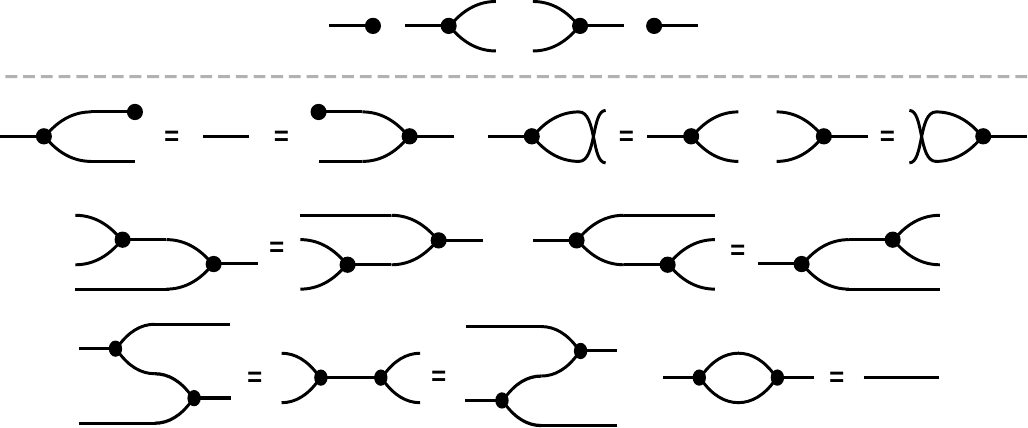}
\end{figure}

Moreover, this structure must be uniform in the sense that the structure on tensor products of objects is given by tensor products of the structure, see \cite{FONG20194746} for more details.

The morphisms of the free hypergraph category over a polygraph are in bijection with multi-pointed hypergraphs in the sense of Habel \cite[Definition 1.3]{habel}, that is, hypergraphs with ``open'' source and target boundaries \cite{10.1145/3502719}. Intuitively, string diagrams in a hypergraph category are hypergraphs.

\section{Details for \Cref{sec:regular}} \label{ax:regular}

\autgram*
\begin{proofsketch}
From the transitions $Δ_γ \subseteq Q^n \times Q^m$ of a "monoidal automaton", we can build a polygraph $ℚ$ by taking a generator $γ_i :  q_1 ⊗ ... ⊗ q_n \to q'_1 ⊗ ... ⊗ q'_m$ for each $((q_1, ..., q_n), (q'_1, ..., q'_m)) \in Δ_γ$. The morphism of polygraphs $\psi : ℚ \to Γ$ simply maps $γ_i$ to $γ$. The reverse is analogous.
\end{proofsketch}

\begin{definition} \label{ax:transition}
  Given a "non-deterministic monoidal automaton" over a "polygraph" $Γ$, we inductively define transition functions $\hat{\delta}_{n,m} : Q^n \times \fpro{Γ}(n,m) \to \mathscr{P}(Q^m)$ over string diagrams in the free monoidal category over $Γ$ with arity $n$ and coarity $m$ as follows:
  \begin{itemize}
  \item For a generator $γ \in Γ$, $\hat{\delta}_{n,m}(q,\gamma) := \delta_{n,m}(\gamma)$,
  \item for identities, $\hat{\delta}_{n,n}(q,\id{}) := \{q\}$,
  \item for a tensor product $s_1⊗s_2$, where $s_1 : n_1 \to m_1, s_2 : n_2 \to m_2$ with $n=n_1+n_2, m=m_1+m_2$ and $q=q_1{++}q_2$, $\hat{\delta}_{n,m}(q, s_1⊗s_2) := \{ p{++}p' \mid p \in \hat{\delta}_{n_1,m_1}(q_1,s_1), \, p' \in \hat{\delta}_{n_2, m_2}(q_2,s_2) \}$,
  \item for a composite $s;s'$, where $s : n \to p, s' : p \to m$, $\hat{δ}_{n,m}(q, s;s') := \{\hat{δ}_{p,m}(q', s') \mid q' \in \hat{δ}_{n,m}(q,s)\}$.
  \end{itemize}
\end{definition}

\section{Proofs omitted from \Cref{sec:cfml}} \label{ax:cfml}

\holessym*
\begin{proof}
  Assume we derived a term $Γ,\Ctx{u},\Ctx{v},Δ ⊢ t$; let us show we could also derive $Γ,\Ctx{v},\Ctx{u},Δ ⊢ t$. We proceed by structural induction, recurring to the first term where the variables $\Ctx{u}$ and $\Ctx{v}$ appeared at two sides of the rule: this rule must have been of the form $t₁ ⨾ t₂$ or $t₁ ⊗ t₂$ for $Γ₁,\Ctx{u},Δ₁ ⊢ t₁$ and $Γ₂,\Ctx{v},Δ₂ ⊢ t₂$, where $Γ ∈ \mathrm{Shuf}(Γ_1;Γ_2)$ and $Δ ∈ \mathrm{Shuf}(Δ_1;Δ_2)$. In that case, we can deduce that
  $Γ,\Ctx{v},\Ctx{u},Δ ∈ \mathrm{Shuf}(Γ₁,\Ctx{u},Δ₁; Γ₂,\Ctx{v},Δ₂)$;
  as a consequence, $Γ,\Ctx{v},\Ctx{u},Δ ⊢ t$ can be derived.
\end{proof}

\holesmonoidal*
\begin{proof}
  We proceed by structural induction. We first note that the three nullary rules of the logic correspond to terms of the free strict monoidal category over the polygraph $𝓟 + \{h_{A,B} \mid A , B ∈ 𝓟_{obj}^{\ast} \}$. The first corresponds to identities, the second corresponds to generators, and the third, when employed with types $A$ and $B$, corresponds to the additional generator $h_{A,B}$. We then note that the two binary rules correspond to sequential and parallel composition, thus obtaining the classical algebraic theory of monoidal terms over the polygraph $𝓟 + \{h_{A,B} \mid A , B ∈ 𝓟_{obj}^{\ast} \}$.
  
  Quotienting over the equations of monoidal categories, as we do when we impose the equations of the theory of diagram contexts, recovers the free strict monoidal category: in a tautological sense, the free strict monoidal category is precisely the one generated by the operations of a monoidal category quotiented by the axioms of a monoidal category. This contrasts sharply with a much more interesting description of the free strict monoidal category: that using string diagrams. As both are exhibited as satisfying the same universal property, they are necessarily equivalent.

  As a particular case, a derivable sequent over the empty context must, by structural induction, avoid any use of the holes. As a consequence of the previous reasoning, it is generated from the polygraph $𝓟$ and it must be a morphism of the free strict monoidal category.

  Finally, the symmetric multifunctor can be described by structural induction: it preserves identities, holes, sequential and parallel compositions, and it sends each monoidal term with no holes $h ∈ \U{\fpro{𝓟}}$ to its derivation under the empty context, $h ∈ \Ctx{𝓟}$.
\end{proof}

\section{Details from \Cref{sec:contour}}

\begin{proposition} \label{prop:comadj}
  The following square of adjunctions commutes.%
\[\begin{tikzcd}[ampersand replacement=\&,cramped]
	{\textsf{PolyGraph}} \& {\textsf{MultiGraph}} \\
	{\textsf{MonCat}} \& {\textsf{MultiCat}}
	\arrow[""{name=0, anchor=center, inner sep=0}, "\raw"', shift right=2, from=1-1, to=1-2]
	\arrow[""{name=1, anchor=center, inner sep=0}, "\fmult{}"', shift right=2, from=1-2, to=2-2]
	\arrow[""{name=2, anchor=center, inner sep=0}, "\fpro{}"', shift right=2, from=1-1, to=2-1]
	\arrow[""{name=3, anchor=center, inner sep=0}, "\raw"', shift right=2, from=2-1, to=2-2]
	\arrow[""{name=4, anchor=center, inner sep=0}, "U"', shift right=2, from=2-1, to=1-1]
	\arrow[""{name=5, anchor=center, inner sep=0}, "\cont"', shift right=2, from=1-2, to=1-1]
	\arrow[""{name=6, anchor=center, inner sep=0}, "U"', shift right=2, from=2-2, to=1-2]
	\arrow[""{name=7, anchor=center, inner sep=0}, "\cont"', shift right=2, from=2-2, to=2-1]
	\arrow["\dashv"{anchor=center}, draw=none, from=2, to=4]
	\arrow["\dashv"{anchor=center, rotate=-90}, draw=none, from=5, to=0]
	\arrow["\dashv"{anchor=center}, draw=none, from=1, to=6]
	\arrow["\dashv"{anchor=center, rotate=-90}, draw=none, from=7, to=3]
\end{tikzcd}\]
\end{proposition}
\begin{proof}
  This follows by unwinding definitions.
\end{proof}

\section{Proofs omitted from \Cref{sec:monoidal-cs}} \label{ax:monoidal-cs-proofs}

\factor*
\begin{proof}
  This is a consequence of the fact that $q^{\ast}$ is full. Given any "diagram context", we argue that we can obtain a (non-unique) "diagram context" of the form of a raw optic
  $$t₁ ⨾ (\id{M_1} \smallotimes \Ctx{x₁} \smallotimes \id{N_1}) ⨾ t₂ ⨾ (\id{M_2} \smallotimes \Ctx{x₂} \smallotimes \id{N_2}) ⨾ ... ⨾ (\id{M_n} \smallotimes \Ctx{xₙ} \smallotimes \id{N_n}) ⨾ t_{n+1}.$$
  Indeed, by structural induction, if the diagram is formed by a hole or a generator, it can be put in raw optic form by adding identities; if the diagram is a composition, we can put both factors in raw optic form and check that their composition is again in raw optic form; if the diagram is a tensoring of two diagrams in raw optic form, we can always apply the interchange law and note that whiskering a raw optic by an object returns again a raw optic.

  It is the case that every map $G → \mathsf{clique}(H)$ arises as a map $G' → H$ for some "multigraph" $G'$ such that $G = \mathsf{clique}(G')$. Combining both facts, we obtain the desired result.
\end{proof}

\rawlang*
\begin{proofsketch}
  A raw representative amounts to choosing a specific ordering of the holes and generators in a "diagram context". By definition (\Cref{lem:factor}), these quotient to the original diagram contexts. In particular, closed derivations quotient to the same element of $ℂ$.
\end{proofsketch}

\uniquemonoidal*
\begin{proof}
  Using the free-forgetful adjunction, the raw representative, $\rawr{\phi}$, determines a unique "multifunctor" $\fmult{G'} \to \raw[ℂ]$.
  Using the adjunction of \Cref{thm:opticadjoint}, this in turn determines a unique "monoidal functor" $\cont{(\fmult{G'})} \to ℂ$.
  Finally, using the commutativity of $\cont$ with $\fmult{}$ (\Cref{prop:comadj}), we obtain $I_\phi : \fpro{(\mathcal{C}G')} \to \mathbb{C}$.
  Explicitly, the action of $I_\phi$ on generators is given by: $A^L \mapsto \pi_1(\rawr{\phi}(A))$, $A^R \mapsto \pi_2(\rawr{\phi}(A))$, $(f,i) \mapsto \pi_i(\rawr{\phi}(f))$, where $\pi$ are projections.
\end{proof}

\bij*
\begin{proof}
  Let $d \in \fmult{G}(;S)$ be a "derivation". We define a family of functions $\{C_X : \fmult{G}(;X) \to \fpro{(\mathcal{C}G)}(X^L,X^R)\}_{X \in G}$ by structural recursion. There are two cases: if $d$ is a generating "operation" $d \in G(;S)$, then $C_S(d) := (d,0) : S^L \to S^R$. Otherwise, $d$ is a composite $(p_1,...,p_n)\comp g $ where $g \in G(X_1,...,X_n;S)$ is a generating "operation" and $p_i \in \fmult{G}(;X_i)$, in which case $C_S(d) := (g,0)\comp C_{X_1}(p_1) \comp (g,1) \comp ... \comp C_{X_n}(p_n) \comp (g,n)$.

  We define functions $\inv{C}_S$ right to left in a similar fashion. Let $c \in \fpro{(\cont{G})}(S^L;S^R)$ be an "optical contour". If $c = (c',0)$ is a generating "sector" then $\inv{C}_S(c) := c'$. Otherwise $c$ is a composite $((g,0) ፡ S^L → M_1 \smallotimes X_1^L \smallotimes N_1) \comp c_1 \comp ((g,1) ፡ M_1 \smallotimes X_1^R \smallotimes N_1 → M_2 \smallotimes X_2^L \smallotimes N_2) \comp ... \comp c_n \comp ((g,n) ፡ M_1 \smallotimes X_1^R \smallotimes N_1 → S^R)$ where $(g,i)$ are generating "sectors" and $c_i \in \fpro{(\cont{G})}(X_i^L, X_i^R)$, in which case $\inv{C}_S(c) := (\inv{C}_{X_1}(c_1), ..., \inv{C}_{X_n}(c_n))\comp g$. It is clear that these functions are mutually inverse and hence form a bijection.
\end{proof}

\main*
\begin{proof}
  By \Cref{lem:raw-lang}, the languages $L(\mathcal{G})$ and $L((\rawr{\phi}, S))$ are equal for any "raw representative" $\rawr{\phi}$ of $\phi$, where $L((\rawr{\phi}, \mathbb{S})) = \rawr{\phi}[\fmult{G'}(;S)]$. It therefore suffices to show that $\rawr{\phi}[\fmult{G'}(;S)] = I_\phi[\fpro{(\cont{G'})}(S^L; S^R)]$. %
    We show the inclusion left to right. Let $d \in \fmult{G'}(;S)$ be a "derivation", and let $C_S(d)$ be the corresponding "optical contour" given by \Cref{lem:bij}. Then by the definition of $I_\phi$ (\Cref{lem:uniquemonoidal}) and $C_S$, we have $I_\phi(C_S(d)) = \rawr{\phi}(d)$.
  We show the inclusion right to left. Let $g \in \fpro{(\cont{G'})}(S^L; S^R)$ be a contour from $S^L$ to $S^R$, and let $\inv{C}_S(g)$ be the corresponding "derivation" given by \Cref{lem:bij}. Then just as before we have $\rawr{\phi}(\inv{C}_S(g)) = I_\phi(g)$.
\end{proof}

\end{document}